\documentclass[letterpaper, 10 pt, conference]{ieeeconf}  

\IEEEoverridecommandlockouts                              

\usepackage{amsfonts}
\usepackage{amsmath}
\usepackage{ifthen}
\usepackage{fancyhdr}
\usepackage{amsfonts}
\usepackage{graphicx}
\usepackage{amsmath}
\usepackage{ifthen}
\usepackage{color}
\usepackage{indentfirst}
\usepackage{epstopdf}
\usepackage{dsfont}
\usepackage{cite}
\usepackage{mdwmath}
\usepackage{texdraw}%
\usepackage{psfrag}%
\usepackage{accents}%
\usepackage{color}%
\usepackage{cite}
\usepackage{graphics}
\usepackage{epigraph}

\usepackage{subcaption}

\usepackage{cleveref}
\DeclareMathOperator{\sgn}{sgn}
\newcommand{\E}{\mathbb{E}}
\newcommand{\Var}{\mathrm{Var}}
\newcommand{\Prob}{\mathbb{P}}
\DeclareMathOperator*{\argmin}{arg\,min}
\makeatletter
\newcommand*{\rom}[1]{\expandafter\@slowromancap\romannumeral #1@}
\makeatother

\newtheorem{remark}{Remark}
\newtheorem{corollary}{Corollary}
\newtheorem{conjecture}{Conjecture}
\newtheorem{assumption}{Assumption}

\title{\LARGE \bf On Remote Estimation with Multiple Communication Channels}
\author{Xiaobin Gao, Emrah Akyol, and Tamer Ba\c{s}ar \thanks{This research was supported in part by NSF under grant CCF 11-11342, and in part by the U.S. Air Force Office of
Scientific Research (AFOSR) MURI grant FA9550-10-1-0573.}\thanks{All authors are with the Coordinated Science Laboratory, University of Illinois at Urbana-Champaign, Urbana, IL 61801; emails: \{xgao16, akyol, basar1\}@illinois.edu}}

\begin{document}

\maketitle
\thispagestyle{empty}
\pagestyle{empty}

\begin{abstract}
This paper considers a sequential estimation and sensor scheduling problem in the presence of multiple communication channels. As opposed to the classical remote estimation problem that involves one perfect (noiseless) channel and one extremely noisy channel (which corresponds to not transmitting the observed state), a  more realistic additive noise channel with fixed power constraint along with a more costly perfect channel is considered. It is shown, via a counter-example, that the common folklore of applying symmetric threshold policy, which is well known to be optimal (for unimodal state densities) in the classical two-channel remote estimation problem, can be suboptimal for the setting considered. Next, in order to make the problem tractable, a side channel which signals the sign of the underlying state is considered. It is shown that, under some technical assumptions, threshold-in-threshold communication scheduling  is  optimal for this setting.  The impact of the presence of a noisy channel is analyzed numerically based on dynamic programming. This numerical analysis uncovers some rather surprising results inheriting known properties from the noisy and noiseless settings.
\end{abstract}

\section{Introduction}
This paper extends the joint sensor scheduling and remote state estimation problems, see e.g., \cite{Imer10, Lipsa11,Nayyar13,GaoACC15}, to a more realistic setting that involves multiple, noisy communication channels.

In \cite{Imer10}, which initiated this line of research, a special case of the problem was considered: Estimate a one-dimensional discrete-time stochastic process distributed independently and identically (i.i.d.) over a decision horizon of length $T$ using only $N \leq T$ measurements. Over the decision horizon of length $T$, the sensor had exactly $N$ opportunities to transmit its observation to the estimator.  The main difference from the work here is that these transmissions were assumed to be {\it error and noise free}.
The transmission decisions that minimize the average estimation error between the  process and its estimate were sought in the class of threshold based strategies and the optimal decision sequence was obtained via dynamic programming. Later, using majorization and related techniques, such threshold based strategies were shown to be optimal for this problem \cite{Lipsa11}.

In a recent prior work\cite{GaoCDC15}, the problem with perfect (noiseless) communication was extended to the noisy channel scenario, i.e., the perfect channel was replaced with a noisy one. Inclusion of noise in the channel poses a significant research challenge, since the sensor now has to encode its message before transmission, and the estimator has to consider this encoding mapping in its estimation mapping. This problem was solved in \cite{GaoCDC15} using the recent results on zero-delay communication \cite{akyol13SourceChannelCoding}. The adversarial zero-delay communication was studied in \cite{akyol2013jamming}, where it was shown that the optimal strategy for an adversarial agent with fixed jamming power is to render the effective channel noise distribution to match that of the source, so that the optimal encoding/decoding mappings are linear. Due to the minimax optimality property of such linear (or affine) mappings\cite{akyol2013jamming}, we pose the problem in an adversarial setting where the noise is generated by a jammer, and we take these communication mappings as affine. 

In this paper, we merge the perfect channel setting, studied in \cite{Imer10,Lipsa11} with the recently studied noisy setting\cite{GaoACC15,GaoCDC15}. An intuitive scheduling policy here  is to use threshold-in-threshold structure since symmetric thresholding has been shown to be optimal, for any unimodal state density, for the noiseless settings\cite{Lipsa11} (under some mild technical conditions). However,  when combined with a noisy channel, we show here that such a policy is no longer optimal and optimal strategy is rather hard to obtain. To facilitate the analysis, we next assume a (perfect) side channel between the encoder and the estimator, over which the sign of the observed state is transmitted. In this setting, in conjunction with some assumptions on the sensing policy and affine encoding-estimating policies, we show optimality of the threshold-in-threshold sensing policy. Armed with this result, we numerically obtain the optimal decision sequence, i.e., the evolution of threshold values in time, via dynamic programming. This numerical analysis demonstrates some rather surprising results inheriting the known properties from the noisy and noiseless settings. For example, the transmitter uses all communication opportunities for the perfect channel, while there might be such opportunities left at the end of the time horizon for communication over the noisy channel.


\section{Problem Formulation}
\label{ProblemFormulation}
\subsection{System Model}
\begin{figure}[h]
\centering
\includegraphics[height=32mm, width = 85mm]{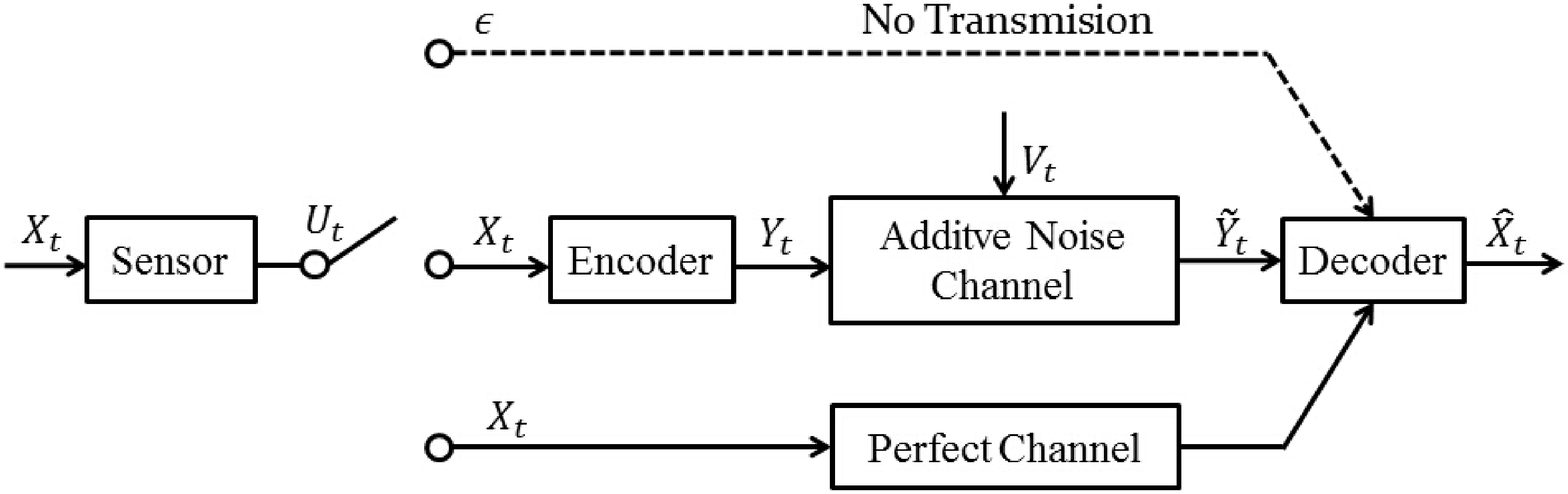}
\caption{System model}
\end{figure}
Consider a discrete time communication scheduling and remote estimation problem over a finite time horizon, i.e., $t=1,2,\ldots,T$. A one-dimensional source process $\{X_t\}$ is an independent identically distributed (i.i.d.) stochastic process with probability density function $p_X$. At time $t$, the sensor observes the state of the source $X_t$. Then, it decides whether and how to transmit its observation to the remote estimator (which is also called ``decoder''). Let $U_t\in\{0,1,2\}$ be the sensor's decision at time $t$. $U_t=0$ means that the sensor chooses not to transmit its observation to the decoder, hence  it sends a free symbol $\epsilon$ to the decoder representing nothing is transmitted. $U_t=1$ means that the sensor chooses to transmit its observation to the decoder over an additive noise channel.  Therefore, the sensor sends $X_t$ to an encoder, which then sends an encoded message, call it $Y_t$, to the communication channel. $Y_t$ is corrupted by an additive channel noise $V_t$. $\{V_t\}$ is a one-dimensional i.i.d. stochastic process with density $p_V$, which is independent of $\{X_t\}$. The encoder has average power constraint, that is,
$$
\E[Y^2] \leq P_T
$$
where $P_T$ is known and constant for all $t$. When $U_t = 2$, sensor chooses to transmit its observation over a perfect channel. Then, the decoder will receive $X_t$. Let $\tilde{Y}_t$ be the message received by decoder at time $t$, we have
$$
\tilde{Y}_t =
\begin{cases}
\epsilon, &\mbox{ if } U_t = 0 \vspace{0.1 cm}
\\ Y_t + V_t, &\mbox{ if } U_t = 1 \vspace{0.1 cm}
\\ X_t, &\mbox{ if } U_t = 2
\end{cases}
$$
After receiving $\tilde{Y}_t$, the decoder generates an estimate on $X_t$, denoted by $\hat{X}_t$. The decoder is charged for squared distortion $(X_t-\hat{X}_t)^2$.

\subsection{Communication Constraints}
We consider the optimization problem under two kinds of communication constraints, separately. In the first scenario, at each time $t$, the sensor is charged for its decision, i.e.,  there is a cost function associated with $U_t$, denoted by $c(U_t)$, such that
$$
c(U_t) =
\begin{cases}
0, &\mbox{ if } U_t = 0 \vspace{0.1 cm}
\\ c_1, &\mbox{ if } U_t = 1 \vspace{0.1 cm}
\\ c_2, &\mbox{ if } U_t = 2
\end{cases}
$$
where $c_1,c_2 \geq0$. $c(U_t)$ is also called communication cost at time $t$. Such kind of communication constraint is called {\it soft constraint}. In the second scenario, the sensor is not charged for transmitting its observations. However, the sensor is restricted to use the noisy channel and the perfect channel for no more than $N_1$ and $N_2$ times, respectively, i.e.,
$$
\sum_{t=1}^T \mathds{1}_{\{U_t = 1\}} \leq N_1, \;\;\; \sum_{t=1}^T \mathds{1}_{\{U_t = 2\}} \leq N_2
$$
where $\mathds{1}_{\{\cdot\}}$ is the indicator function, and $N_1$, $N_2$ are positive integers. Such kind of communication constraint is called {\it hard constraint}.

\subsection{Decision Strategies}
Assume that at time $t$, the sensor has memory of all its measurements by $t$, denoted by $X_{1:t}$, and all the decisions it has made by $t-1$, denoted by $U_{1:t-1}$. The sensor makes decision $U_t$ based on its current information $(X_{1:t},U_{1:t-1})$, that is,
$$
U_t=f_t(X_{1:t},U_{1:t-1})
$$
where $f_t$ is the sensor scheduling policy at time $t$ and $\textbf{f}=\{f_1,f_2,\ldots,f_T\}$ is the sensor scheduling strategy.

Assume that at time $t$, no matter whether and how the sensor decides to transmit the source output, it always transmits its decision $U_t$ to the encoder. Let $\tilde{X}_t$ be the message received by the encoder at time $t$; then,
$$
\tilde{X}_t =
\begin{cases}
(X_t,U_t), &\mbox{ if } U_t = 1 \vspace{0.1 cm}
\\ U_t, &\mbox{ otherwise }
\end{cases}
$$
Denote by $\tilde{X}_{1:t}$ the messages received by the encoder up to time $t$. Similar to the above, we assume that the encoder has memory on $\tilde{X}_{1:t}$, and all the encoded messages it has sent to the communication channel by $t-1$, denoted by $Y_{1:t-1}$. The encoder generates the encoded message $Y_t$ based on its current information $(\tilde{X}_{1:t},Y_{1:t-1})$, that is,
$$
Y_t=g_t(\tilde{X}_{1:t},Y_{1:t-1}),
$$
where $g_t$ is the encoding policy at time $t$ and $\textbf{g}=\{g_1,g_2,\ldots,g_T\}$ is the encoding strategy.

Assume that the decoder can deduce $U_t$ from $\tilde{Y}_t$. Furthermore, it is assumed that at time $t$, the decoder has memory on all the messages received from communication channels by $t$, denoted by $\tilde{Y}_{1:t}$. The decoder produces the estimate $\hat{X}_t$ based on its current information $\tilde{Y}_{1:t}$, namely,
$$
\hat{X}_t=h_t(\tilde{Y}_{1:t}),
$$
where $h_t$ is the decoding policy at time $t$ and $\textbf{h}=\{h_1,h_2,\ldots,h_T\}$ is the decoding strategy.
\begin{remark} \label{Shared information} Although we do not assume that the decoder has memory on $\hat{X}_{1:t-1}$, yet it can deduce them from $\tilde{Y}_{1:t-1}$ and $\{h_1,h_2,\ldots,h_{t-1}\}$. \end{remark}

\subsection{Optimization Problems}
Consider the setting described above, with the time horizon $T$, the probability density functions $p_X$ and $p_V$, and the power constraint $P_T$ as given.

Optimization problem with {\it soft constraint}:  Given the communication cost function $c(\cdot)$, determine  $(\textbf{f},\textbf{g},\textbf{h})$  that minimize
$$
J(\textbf{f},\textbf{g},\textbf{h})= \E \left\{\sum_{t=1}^T c(U_t)+{(X_t-\hat{X}_t)}^2\right\}.
$$

Optimization problem with {\it hard constraint}: Given the number of transmission opportunities $N_1$ and $N_2$, determine $(\textbf{f},\textbf{g},\textbf{h})$ that minimize
$$
J(\textbf{f},\textbf{g},\textbf{h})= \mathbb{E} \left\{\sum_{t=1}^T {(X_t-\hat{X}_t)}^2\right\}.
$$

\section{Preliminaries}
\label{prior work}
\subsection{Problem with Perfect Channel}
The communication scheduling and remote estimation problems with one perfect communication channel and soft/hard constraints have been studied in \cite{Imer10,Lipsa11,Nayyar13}. In this prior work, both i.i.d. source and Markov source were considered. In the case of i.i.d. source, it was assumed that the source density is symmetric and unimodal around $0$, namely,
$$
\begin{array}{lcl}
p_X(x) &=& p_X(-x), \;\;\;\forall\;x\in\mathbb{R} \vspace{0.1 cm}
\\ p_X(a) &\geq& p_X(b), \;\;\;\;\mbox{if } |a|\leq|b|
\end{array}
$$
One of the distortion metrics considered was the squared error. With the above assumptions, it was shown that the optimal communication scheduling policy at time $t$ is symmetric threshold-based, and the optimal estimation policy is also symmetric, that is,
$$
\begin{array}{lcl}
f_t(X_t) =
\begin{cases}
2, &\mbox{ if } |X_t| > \beta_t \vspace{0.1 cm}
\\ 0, &\mbox{ if } |X_t| \leq \beta_t
\end{cases}
,h_t(\tilde{Y}_t) =
\begin{cases}
\tilde{Y}_t, &\mbox{ if } U_t = 2 \vspace{0.1 cm}
\\ 0, &\mbox{ if } U_t = 0
\end{cases}
\end{array}
$$
where $\beta_t$ is the threshold at time $t$. In the problem with soft constraint, $\beta_t$ depends only on communication cost, which is independent of time. In the problem with hard constraint, $\beta_t$ depends on time $t$ and the number of communication opportunities left at time $t$, denoted by $E_t$, where
$$
E_t = N_2 -  \sum_{i=1}^{t-1} \mathds{1}_{\{U_i = 2\}}
$$
Furthermore, $\beta_t$ can be computed via dynamic programming (see \cite{Imer10} for details).
\subsection{Zero Delay Communication and Jamming}
In \cite{akyol13SourceChannelCoding}, the following problem was considered: an encoder wants to make a one-shot transmission sending an input signal $X$ to the decoder. The communication channel has a zero mean additive channel noise $V$, which is independent of $X$. Since the channel is noisy, the encoder first encodes the input signal $X$ according to some encoding policy $g$, and then sends the encoded message $Y$ to the communication channel. The encoder is assumed to have average encoding power constraint. The decoder receives the noise-corrupted message $Y+V$, denoted by $\tilde{Y}$, and then generates an estimate of $X$, denoted by $\hat{X}$, according to some decoding policy $h$. The optimization problem is to design the encoding and decoding policies $(g,h)$ to minimize the mean squared estimation error $\E[(X-\hat{X})^2]$ subject to encoding power constraint $\E[Y^2]\leq P_T$. It has been shown that once the characteristic functions of $X$ and $V$ satisfy so-called {\it matching conditions}, the optimal encoding and decoding policies are  affine as follows:
$$
\begin{array}{lcl}
g(X) &=& \alpha \cdot (X - \E[X]) \vspace{0.1cm}
\\ h(\tilde{Y}) &=& \frac{1}{\alpha}\frac{\gamma}{\gamma+1}\tilde{Y}+\E[X]
\end{array}
$$
where $\gamma:=\frac{P_T}{\sigma_V^2}$ is the signal to noise ratio (SNR), $\alpha=\sqrt{\frac{P_T}{\Var(X)}}$, $\Var(X)$ is the variance of $X$. Furthermore, minimum mean squared error is $\frac{\Var(X)}{1+\gamma}$.

Later in \cite{akyol2013jamming}, a jamming problem was considered where the communication channel noise is generated by an adversary, and it was shown that the affine encoding/decoding policies are minimax.

\subsection{Problem with Noisy Channel}
In the work of \cite{GaoACC15,GaoCDC15}, similar problems with only one additive noise channel with soft/hard constraints were analyzed. It was shown that if the source and the channel noise are i.i.d., and the communication cost function, the distortion metric, and the encoding power are time invariant, then the optimization problem over a finite time horizon with soft constraint collapses to a one-stage problem \cite[Theorem 2]{GaoACC15}, and the optimization problem with hard constraint can be converted to a one-stage optimization problem with soft constraint \cite[Theorems 2 and 3]{GaoCDC15}.


\section{The Problem with Soft Constraint}
\subsection{Conjecture and Corollary}
By an argument similar to that in \cite[Theorem 2]{GaoACC15}, the optimization problem with i.i.d. source and soft communication constraint collapses to a one-stage optimization problem. Hence for simplicity, we henceforth suppress the subscript for time. We make the following assumptions on the optimization problem.
\begin{assumption}
\label{assumption on the problem}
 The source density $p_X$ is symmetric and unimodal around zero.
\end{assumption}
\begin{assumption}
 The communication channel noise $V$ has zero mean, and fixed variance, denoted by $\sigma_V^2$.
\end{assumption}
\begin{assumption}
 The encoder and decoder are restricted to apply affine policies, namely
$$
\begin{array}{lcl}
g(X) &=& \alpha \cdot (X - \E[X|U=1]) \vspace{0.1cm}
\\ h(\tilde{Y}) &=& \frac{1}{\alpha}\frac{\gamma}{\gamma+1}\tilde{Y}+\E[X|U=1]
\end{array}
$$
where $\gamma:=\frac{P_T}{\sigma_V^2}$ is the signal to noise ratio, $\alpha=\sqrt{\frac{P_T}{\Var(X|U=1)}}$. $\Var(X|U=1)$ is the variance of $X$ condition that the sensor transmits the source output over the noisy channel.
\end{assumption}

The first assumption is standard. The second and third assumptions are consequences of  the jamming setting (that is, with worst-case approach, see \cite{akyol2013jamming} for details). Since the source is symmetric around zero, and the distortion metric is the squared error, which is also symmetric around zero, it is intuitive to conjecture that the communication scheduling policy is symmetric around zero. Also note that the problem with one perfect communication channel admits an optimal scheduling policy which is symmetric around zero.

 \begin{conjecture}
\label{conjecture on symmetric communication policy}
The optimal communication scheduling policy satisfies$$
f(x) = f(-x)\;\;\;\forall\; x\in\mathbb{R}
$$
\end{conjecture}

The following corollary is a sequence of Conjecture \ref{conjecture on symmetric communication policy}.
\begin{corollary}
\label{corollary with conjecture on non side channel problem}
If Conjecture \ref{conjecture on symmetric communication policy} holds,  the optimal  scheduling policy is of the threshold-in-threshold type:
\begin{equation}
\label{threshold in threshold expression}
f(x) =
\begin{cases}
0, & \mbox{ if } |x| \leq \beta_1 \vspace{0.1 cm}
\\ 1, & \mbox{ if } \beta_1 < |x| \leq \beta_2 \vspace{0.1 cm}
\\ 2, & \mbox{ if } |x| > \beta_2
\end{cases}
\end{equation}
where $\beta_1$ and $\beta_2$ are called thresholds, and $0 \leq \beta_1 \leq \beta_2$.
\end{corollary}
\begin{proof}
Let $\mathcal{T}^f_{0}$, $\mathcal{T}^f_1$, $\mathcal{T}^f_2$ be the {\it non-transmission region}, the {\it noisy transmission region}, and the {\it perfect transmission region},
respectively, according to communication policy $f$, i.e.,
$$
\mathcal{T}^f_{i} := \{x\in\mathbb{R} | f(x) = i\},\;\;i\in\{0,1,2\}
$$
Then, the  conjecture states that $\mathcal{T}^f_{0}$, $\mathcal{T}^f_1$, $\mathcal{T}^f_2$ are symmetric around zero.

When $x\in \mathcal{T}^f_{0}$, the sensor does not send anything to the decoder but the free symbol $\epsilon$. Then, the decoder knows that $X\in \mathcal{T}^f_{0}$. Hence, $\hat{X}=\E[X| X\in \mathcal{T}^f_{0}]$. By Conjecture \ref{conjecture on symmetric communication policy}, we have $\hat{X}=\E[X| X\in \mathcal{T}^f_{0}]=0$.

When $x\in \mathcal{T}^f_{2}$, the sensor chooses to transmit its observation over the perfect channel. Hence, the optimal decoder is to report the received message, that is, $\hat{X} = x$.

When $x\in \mathcal{T}^f_{1}$, the sensor chooses to transmit its observation over the noisy channel. By Conjecture \ref{conjecture on symmetric communication policy}, we have
$$
\E\big[X|U=1\big] = \E\big[X\big \vert X\in \mathcal{T}^f_{1}] = 0
$$
where the second equality is due to the fact that $p_X$ and $\mathcal{T}^f_{1}$ are symmetric around zero. Hence, we have
$$
g(X) = \alpha \cdot X, \;\;\; h(\tilde{Y}) = \frac{1}{\alpha}\frac{\gamma}{\gamma+1}\tilde{Y}
$$
where $\gamma = \frac{P_T}{\sigma_V^2}$ is given, and $\alpha$ depends on the choices of $f$ and $p_X$. We are done with the proof if we can show that given any $\alpha>0$,  \eqref{threshold in threshold expression} is satisfied. Suppose the sensor observes the realization of $X=x$. Let $J_0(x)$, $J_1(x)$, and $J_2(x)$ be the total cost functions corresponding to $U=0$, $U=1$, and $U=2$, respectively. Then we have,
$$
\begin{array}{lcl}
J_0(x) &=& x^2 \vspace{0.1 cm}
\\ J_1(x) &=& c_1 + \E[(x-\hat{X})^2] \vspace{0.1 cm}
\\ &=& c_1 + \E \big[(x-\frac{1}{\alpha}\frac{\gamma}{\gamma+1}(\alpha x+V))^2\big] \vspace{0.1 cm}
\\ &=& c_1 + \frac{1}{(\gamma+1)^2} \E\big[(x-\frac{\gamma}{\alpha}V)^2\big] \vspace{0.1 cm}
\\ &=& c_1 + \frac{1}{(\gamma+1)^2}\cdot x^2 + \frac{\gamma^2}{\alpha^2(\gamma+1)^2}\cdot\sigma_V^2 \label{second} \vspace{0.1 cm}
\\ J_2(x) &=& c_2
\end{array}
$$
Then,
$$
f(x) = \underset{i\in\{0,1,2\}}{\argmin}\; J_i(x)
$$
$J_0(x)$, $J_1(x)$, and $J_2(x)$ are symmetric around zero, hence we only need to consider the case when $x\geq 0$. Since $\frac{1}{(1+\gamma)^2}<1$, it is easy to check that there exist $\beta_{01}$ and $\beta_{02}$ such that
$$
\begin{array}{lcl}
J_0(x) \leq J_1(x),\; \mbox{iff} \;x\in[0,\beta_{01}] \vspace{0.1 cm}
\\ J_0(x) \leq J_2(x),\; \mbox{iff} \;x\in[0,\beta_{02}]
\end{array}
$$
Define $\beta_1:=\min\{\beta_{01},\beta_{02}\}$ and we have
\begin{equation}
\label{without side channel ball in ball 1}
\begin{array}{lcl}
&\;& J_0(x) \leq \min\{J_1(x),J_2(x)\}, \mbox{iff} \;x\in[0,\beta_1] \vspace{0.1 cm}
\\ &\Rightarrow&  f(x)=0 , \mbox{iff} \;x\in[0,\beta_1]
\end{array}
\end{equation}
Therefore, when computing $f(x)$ for $x\in(\beta_1,\infty)$, we only need to compare $J_1(x)$ with $J_2(x)$. Consider $J_1(x)$ and $J_2(x)$, since $J_1(x)$ is a parabolic function of $x$ and $J_2(x)$ is constant in $x$, either of the following cases occurs:

Case \rom{1}: $c_2<c_1+\frac{\gamma^2}{\alpha^2(\gamma+1)^2}\cdot\sigma_V^2$. Then
\begin{equation}
\label{without side channel ball in ball 2}
J_1(x)>J_2(x),\;\forall\; x>0 \Rightarrow f(x) = 2,\;\forall\; x\in(\beta_1,\infty)
\end{equation}

Case \rom{2}: $c_2\geq c_1+\frac{\gamma^2}{\alpha^2(\gamma+1)^2}\cdot\sigma_V^2$. It can be checked that there exists one threshold, call it $\beta_2$, such that $J_1(x) \leq J_2(x)$ if and only if $x\in[0,\beta_2]$. If $\beta_2<\beta_1$,
\begin{equation}
\label{without side channel ball in ball 3}
J_1(x)>J_2(x), x\in(\beta_1,\infty) \Rightarrow f(x) = 2, x\in(\beta_1,\infty)
\end{equation}
If $\beta_2\geq\beta_1$,
\begin{equation}
\begin{array}{l}
J_1(x)\leq J_2(x), x\in(\beta_1,\beta_2] \Rightarrow f(x) =1,  x\in(\beta_1,\beta_2]  \vspace{0.1 cm}
\\J_1(x)>J_2(x), x\in(\beta_2,\infty)\Rightarrow f(x) =2, x\in(\beta_2,\infty)
\end{array}
\label{without side channel ball in ball 4}
\end{equation}
Combining \eqref{without side channel ball in ball 1} through \eqref{without side channel ball in ball 4}, we conclude that the optimal communication scheduling policy $f(x)$ has the expression of \eqref{threshold in threshold expression}. Note that \eqref{without side channel ball in ball 1}-\eqref{without side channel ball in ball 2} and \eqref{without side channel ball in ball 1} and \eqref{without side channel ball in ball 3} are the special cases of \eqref{threshold in threshold expression} where $\beta_2=\beta_1$.
\end{proof}

Although Conjecture \ref{conjecture on symmetric communication policy}  and  Corollary \ref{corollary with conjecture on non side channel problem} seem very intuitive at first glance, the following counter example renders them not valid from the point of global optimality.

{\it Counter example}: Consider
$$
p_X(x) = \frac{1}{2L}, \;x\in[-L,L]
$$
which  is symmetric and unimodal. Assume  $c_1<c_2$. By Corollary \ref{corollary with conjecture on non side channel problem}, the optimal communication scheduling policy, denoted by $f^\ast$, is of the threshold-in-threshold type, which is described in \eqref{threshold in threshold expression} with thresholds $0<\beta_1^\ast<\beta_2^\ast$. By the proof of Corollary \ref{corollary with conjecture on non side channel problem}, $\beta_2^\ast$ is the separating point where $J_1(x) \leq J_2(x)$ if and only if $x\in[0,\beta^\ast_2]$. Consider $J_1(x)-J_2(x)$:
$$
\begin{array}{lcl}
J_1(x)-J_2(x) &=& c_1-c_2 + \frac{1}{(\gamma+1)^2}\cdot x^2 + \frac{\gamma^2}{\alpha^2(\gamma+1)^2}\cdot\sigma_V^2 \vspace{0.1 cm}
\\ &\geq & c_1-c_2 + \frac{1}{(\gamma+1)^2}\cdot x^2 \vspace{0.1 cm}
\\ &> & 0,\;\mbox{if}\;x>\sqrt{c_2-c_1}\cdot(\gamma+1)
\end{array}
$$
which implies that $\beta^\ast_2<\sqrt{c_2-c_1}\cdot(\gamma+1)$. Hence, by choosing $\sqrt{c_2-c_1}\cdot(\gamma+1)\ll L$, we have $\beta_2^\ast\ll L$. Denote by $J(f)$, the expected total cost if the sensor applies communication scheduling policy $f$. Then, we have
$$
\begin{array}{lcl}
\;\;\; J(f) = \E\big[c(U)+(X-\hat{X})^2\big] \vspace{0.1 cm}
\\ = \sum_{i=0}^2 \E\big[c(U)+(X-\hat{X})^2\big| X\in \mathcal{T}^f_{i}] \cdot \Prob(X \in \mathcal{T}^f_{i}) \vspace{0.1 cm}
\\ = \E[(X-\hat{X})^2| X\in \mathcal{T}^f_{0}] \cdot \Prob(X\in \mathcal{T}^f_{0}) + c_1 \cdot  \Prob(X\in \mathcal{T}^f_{1}) \vspace{0.1 cm}
\\ +\; \E[(X-\hat{X})^2| X\in \mathcal{T}^f_{1}] \cdot \Prob(X\in \mathcal{T}^f_{1}) + c_2 \cdot  \Prob(X\in \mathcal{T}^f_{2})
\end{array}
$$
Recall that when $X\in \mathcal{T}^f_{0}$ , $\hat{X} = \E[X| X\in \mathcal{T}^f_{0}]$. Hence, $\E\big[(X-\hat{X})^2\big| X\in \mathcal{T}^f_{0}]$ = $\Var(X| X\in \mathcal{T}^f_{0})$. Furthermore, by the results from \cite{akyol13SourceChannelCoding} discussed in section \ref{prior work}, $\E[(X-\hat{X})^2| X\in \mathcal{T}^f_{1}]=\frac{1}{\gamma+1}\Var(X| X\in \mathcal{T}^f_{1})$. Hence,
\begin{equation}
\label{expression of total cost function}
\begin{array}{lcl}
\;\;\; J(f) = \Var(X| X\in \mathcal{T}^f_{0}) \cdot \Prob(X\in \mathcal{T}^f_{0}) + c_1  \Prob(X\in \mathcal{T}^f_{1}) \vspace{0.2 cm}
\\ +\; \frac{1}{\gamma+1}\cdot\Var(X| X\in \mathcal{T}^f_{1}) \cdot \Prob(X\in \mathcal{T}^f_{1}) + c_2  \Prob(X\in \mathcal{T}^f_{2})
\end{array}
\end{equation}
Next, consider $f^\ast$; then $\mathcal{T}^{f^\ast}_{0}$ $=$ $[-\beta_1^\ast,\beta_1^\ast]$, $\mathcal{T}^{f^\ast}_{1} = [-\beta_2^\ast,-\beta_1^\ast)$ $\bigcup$ $(\beta_1^\ast,\beta_2^\ast]$, and $\mathcal{T}^{f^\ast}_{2} = [-L,-\beta_2^\ast)$ $\bigcup$ $(\beta_2^\ast,L]$. We now construct another communication scheduling policy $f^\prime$ as follows
$$
\begin{array}{lcl}
\mathcal{T}^{f^\prime}_{0} = \mathcal{T}^{f^\ast}_{0}, \mathcal{T}^{f^\prime}_{1} = (\beta_1^\ast,\beta_2^\ast] \;\bigcup\; (\beta_2^\ast,2\beta_2^\ast-\beta_1^\ast] \vspace{0.1 cm}
\\ \mathcal{T}^{f^\prime}_{2} = [-L,-\beta_1^\ast)\;\bigcup\;(2\beta_2^\ast-\beta_1^\ast,L]
\end{array}
$$
One can see that we shifted part of the noisy transmission region to make it connected, as illustrated in Fig. \ref{fig:counterExample}.
\begin{figure}[h]
\centering
\includegraphics[height=40mm, width = 85mm]{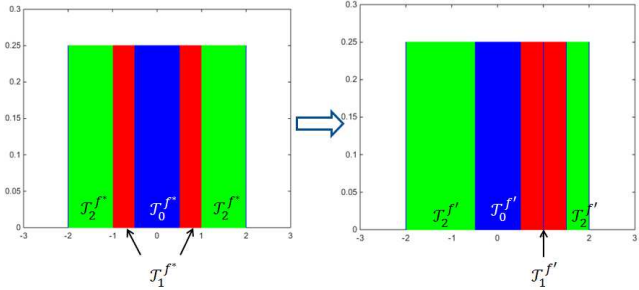}
\caption{The counter example}
\label{fig:counterExample}
\end{figure}
Since the source has uniform distribution, $\Prob(X\in \mathcal{T}^{f^\prime}_{1})$ $=\Prob(X\in \mathcal{T}^{f^\ast}_{1})$, and $\Prob(X\in \mathcal{T}^{f^\prime}_{2})$ $=\Prob(X\in \mathcal{T}^{f^\ast}_{2})$. Furthermore, the non-transmission region, $\mathcal{T}^{f^\ast}_0$, is not changed. Hence,
$$
\begin{array}{lcl}
\;\;\;J(f^\prime) - J(f^\ast) \vspace{0.2 cm}
\\ = \dfrac{\Prob(X\in \mathcal{T}^{f^\prime}_{1})}{\gamma+1}\cdot \Big(\Var(X| X\in \mathcal{T}^{f^\prime}_{1})-\Var(X| X\in \mathcal{T}^{f^\ast}_{1})\Big)
\end{array}
$$
Since $p_X$ is uniform, $\mathcal{T}^{f^\prime}_1$ and $\mathcal{T}^{f^\ast}_1$ have the same Lebesgue measure, and $\mathcal{T}^{f^\prime}_1$ is connected, we conclude that $\Var(X| X\in \mathcal{T}^{f^\prime}_{1})$ $<\Var(X| X\in \mathcal{T}^{f^\ast}_{1})$ and thus, $J(f^\prime)<J(f^\ast)$ which generates a contradiction. In the special case where $\beta^\ast_1 = \beta^\ast_2$, one can also come up with a counter example by replacing part of perfect transmission region $\mathcal{T}_2^{f^\ast}$ by noisy transmission region $\mathcal{T}_1^{f^\prime}$.

\begin{remark}
The counter example above shows that noisy transmission region in the symmetric communication policy may be disconnected. As discussed in section \ref{prior work}, MMSE of the zero delay communication problem is proportional to  $\Var(X| X\in \mathcal{T}^{f}_{1})$. Splitted noisy transmission region results in large $\Var(X| X\in \mathcal{T}^{f}_{1})$, and thus does not take full advantage of the noisy channel.
\end{remark}

In order to have symmetric noisy transmission region to render the problem tractable, we further assume the existence of side channel.
\subsection{Modified Problem}
Assume that there exists a perfect side channel between the encoder and the decoder. When transmitting the encoded message $Y$, the encoder also sends the sign of the source, denoted by $S$, to the decoder over the side channel. 
The decoder generates the estimate $\hat{X}$ based on the received messages $(\tilde{Y},S)$. Note that side-channel will not be used if the sensor chooses not to transmit the source output or to transmit it over the perfect channel. Hence, we have
$$
S =
\begin{cases}
\sgn(X), &\mbox{ if } U=1 \vspace{0.1 cm}
\\ \epsilon, &\mbox{otherwise}
\end{cases}
$$
where $\epsilon$ stands for nothing being transmitted. Note that the side channel enables using  different encoding/decoding policies for the positive and negative input signal. Hence, we need to modify Assumption 3 (we keep Assumptions 1 and 2).
\begin{assumption}
 The encoder and the decoder are restricted to apply piecewise affine policies:
$$
\begin{array}{lcl}
g(X,S) &=& S \cdot \alpha \cdot \left(X - \E\left[X|U=1,S\right] \right) \vspace{0.1cm}
\\ h(\tilde{Y},S) &=& S \cdot \frac{1}{\alpha}\frac{\gamma}{\gamma+1}\tilde{Y}+\E\left[X|U=1,S\right]
\end{array}
$$
where $\gamma:=\frac{P_T}{\sigma_V^2}$ is the signal to noise ratio, $\alpha=\sqrt{\frac{P_T}{\Var(X|U=1,S)}}$. $\Var(X|U=1,S)$ is the conditional variance.
\end{assumption}


Let $\mathcal{T}^f_{1+}$, $\mathcal{T}^f_{1-}$ be the {\it positive noisy transmission region} and the {\it negative noisy transmission region}, respectively,
according to communication policy $f$, i.e.,
$$
\mathcal{T}^f_{1+} = \{x>0|f(x) = 1\},\;\;\mathcal{T}^f_{1-} = \{x<0|f(x) = 1\}
$$
Note that even under the assumption of symmetric communication scheduling policy, we may still have connected positive/negative noisy transmission regions, which would result in small conditional variance. Therefore, we still have Conjecture \ref{conjecture on symmetric communication policy}. Then we can show that Corollary \ref{corollary with conjecture on non side channel problem} still holds based on Conjecture \ref{conjecture on symmetric communication policy}.

\begin{corollary}
\label{corollary for modified problem}
In the modified problem, if the sensor is restricted to apply symmetric communication scheduling policy described in Conjecture \ref{conjecture on symmetric communication policy}, then the optimal communication scheduling policy is of the threshold-in-threshold type described by \eqref{threshold in threshold expression} in Corollary \ref{corollary with conjecture on non side channel problem}.
\end{corollary}

\begin{proof}
We use an argument similar to that in the proof of Corollary \ref{corollary with conjecture on non side channel problem}. When $X=x\in \mathcal{T}^f_{0}$, $\hat{X}=0$. Furthermore, by the symmetry of $f$, we have $\mathcal{T}^f_{1+} = -\mathcal{T}^f_{1-}$. Since $p_X$ is symmetric, we have
$$
\begin{array}{rcr}
\E\left[X|U=1,S=+1\right] &=& \E\left[X|X\in\mathcal{T}^f_{1+}\right] \vspace{0.1 cm}
\\ =  -\E\left[X|U=1,S=-1\right] &=& -\E\left[X|X\in\mathcal{T}^f_{1-}\right]
\end{array}
$$
Let $b:=\E[X|U=1,S=+1]$. Then $\E[X|U=1,S]$$=Sb$. $\gamma$ is known, while $\alpha$ and $b$ depend on the choice of $f$ and $p_X$. For any $\alpha,b>0$, any realization of source output $X=x$, and the corresponding realization of $S$, denoted by $S=s=\sgn(x)$, we can compute the total cost functions $J_0(x)$, $J_1(x)$ and $J_2(x)$ as follows:
\begin{equation}
\label{three total cost functions}
\begin{array}{lcl}
J_0(x) &=& x^2 \vspace{0.1 cm}
\\ J_1(x) &=& c_1 + \E[(x-\hat{X})^2] \vspace{0.1 cm}
\\ &=& c_1 + \E \big[(x-\frac{1}{\alpha}\frac{\gamma}{\gamma+1}(\alpha x-\alpha s b+sV)-sb)^2\big] \vspace{0.1 cm}
\\ &=& c_1 + \frac{1}{(\gamma+1)^2} \E\big[(x-sb-\frac{\gamma}{\alpha}sV)^2\big] \vspace{0.1 cm}
\\ &=& c_1 + \frac{1}{(\gamma+1)^2}\cdot (x-sb)^2 + \frac{\gamma^2}{\alpha^2(\gamma+1)^2}\cdot\sigma_V^2 \vspace{0.1 cm}
\\ &=& c_1 + \frac{1}{(\gamma+1)^2}\cdot (|x|-b)^2 + \frac{\gamma^2}{\alpha^2(\gamma+1)^2}\cdot\sigma_V^2 \vspace{0.1 cm}
\\ J_2(x) &=& c_2
\end{array}
\end{equation}
where the second last equality is due to the fact that $x = s|x|$ and $s^2=1$. Since $J_0(x)$, $J_1(x)$, and $J_2(x)$ are even functions of $x$, we only need to consider the case where $x\geq 0$. It is easy to see that there exists $\beta_{02}$ such that
$$
J_0(x) \leq J_2(x),\; \mbox{iff} \;x\in[0,\beta_{02}]
$$
$J_0(x)$ and $J_1(x)$ are quadratic functions, $J_0(0)-J_1(0)<0$. Furthermore,
$$
\frac{d}{dx}\Big( J_0(x)-J_1(x)\Big) = \frac{2\gamma^2+4\gamma}{(\gamma+1)^2}x + \frac{2b}{(\gamma+1)^2} > 0,\;\forall\; x > 0
$$
Hence, there exists $\beta_{01}$ such that
$$
J_0(x) \leq J_1(x),\; \mbox{iff} \;x\in[0,\beta_{01}]
$$
Therefore,
$$
\begin{array}{lcl}
&\;& J_0(x) \leq \min\{J_1(x),J_2(x)\}, \mbox{iff} \;x\in[0,\beta_1] \vspace{0.1 cm}
\\ &\Rightarrow&  f(x)=0 , \mbox{iff} \;x\in[0,\beta_1]
\end{array}
$$
where $\beta_1 = \min\{\beta_{01},\beta_{02}\}$. Hence when considering $x\in(\beta_1,\infty)$, we only need to compare $J_1(x)$ with $J_2(x)$. $J_1(x)$ is a porabolic opening upward, and $J_2(x)$ is constant. Therefore when $x\in(\beta_1,\infty)$ there are three possibilities. Case \rom{1}: $J_1(x)$ and $J_2(x)$ do not intersect, which implies
$$
J_1(x) > J_2(x), x\in(\beta_1,\infty) \Rightarrow f(x) = 2, x\in(\beta_1,\infty)
$$
Case \rom{2}: $J_1(x)$ and $J_2(x)$ intersect once at $\beta_{12,r}$, which implies
$$
\begin{array}{lcl}
J_1(x) \leq J_2(x)\Rightarrow f(x) = 1, x\in(\beta_1,\beta_{12,r}] \vspace{0.1 cm}
\\ J_1(x) > J_2(x)\Rightarrow f(x) = 2, x\in(\beta_{12,r},\infty)
\end{array}
$$
Case \rom{3}: $J_1(x)$ and $J_2(x)$ intersect twice at $\beta_{12,l}$ and $\beta_{12,r}$, which implies
\begin{equation}
\label{ball in ball in ball}
\begin{array}{lcl}
J_1(x) \geq J_2(x)\Rightarrow f(x) = 2, x\in(\beta_1,\beta_{12,l}] \vspace{0.1 cm}
\\ J_1(x) \leq J_2(x)\Rightarrow f(x) = 1, x\in(\beta_{12,l},\beta_{12,r}] \vspace{0.1 cm}
\\ J_1(x) > J_2(x)\Rightarrow f(x) = 2, x\in(\beta_{12,r},\infty)
\end{array}
\end{equation}
In cases \rom{1} and \rom{2}, the optimal communication scheduling policies are of the threshold-in-threshold type, while that conclusion does not hold for case \rom{3}. We now show that for any symmetric communication scheduling policy $f$ in the form of \eqref{ball in ball in ball}, we can come up with a threshold-in-threshold type policy achieving no higher cost. Consider a communication scheduling policy $f$ described by \eqref{ball in ball in ball} with thresholds $\beta_1$, $\beta_{12,l}$, and $\beta_{12,r}$. Then,
$$
\begin{array}{lcl}
\mathcal{T}^f_0 &=& [-\beta_1,\beta_1], \vspace{0.1 cm}
\\ \mathcal{T}^f_{1+} &=& (\beta_{12,l},\beta_{12,r}] , \;\;\;\mathcal{T}^f_{1-} = [-\beta_{12,r},-\beta_{12,l}) , \vspace{0.1 cm}
\\  \mathcal{T}^f_2 &=& (-\infty,-\beta_{12,r})\;\bigcup\;[-\beta_{12,l},-\beta_1) \vspace{0.1 cm}
\\ &\;&\;\;\;\bigcup\;  (\beta_1,\beta_{12,l}]\; \bigcup \;(\beta_{12,r},\infty)
\end{array}
$$
Similar to \eqref{expression of total cost function}, the total expected cost by applying $f$ can be computed as
\begin{equation}
\label{expression of J(f)}
\begin{array}{lcl}
\;\;\; J(f) = \Var(X| X\in \mathcal{T}^f_{0}) \cdot \Prob(X\in \mathcal{T}^f_{0}) + c_1  \Prob(X\in \mathcal{T}^f_{1+}) \vspace{0.2 cm}
\\ +\; \frac{1}{\gamma+1}\cdot\Var(X| X\in \mathcal{T}^f_{1+}) \cdot \Prob(X\in \mathcal{T}^f_{1+}) + c_1  \Prob(X\in \mathcal{T}^f_{1-}) \vspace{0.2 cm}
\\ +\; \frac{1}{\gamma+1}\cdot\Var(X| X\in \mathcal{T}^f_{1-}) \cdot \Prob(X\in \mathcal{T}^f_{1-}) +  c_2  \Prob(X\in \mathcal{T}^f_{2})
\end{array}
\end{equation}
Based on $f$, construct $f'$ as follows:
$$
\begin{array}{lcl}
\mathcal{T}^{f'}_0 &=& [-\beta_1,\beta_1], \vspace{0.1 cm}
\\ \mathcal{T}^{f'}_{1+} &=& (\beta_1,\beta_2^\prime] , \;\;\;\mathcal{T}^{f'}_{1-} = [-\beta_2^\prime,-\beta_1) , \vspace{0.1 cm}
\\  \mathcal{T}^{f'}_2 &=& (-\infty,-\beta_2^\prime)\; \bigcup \;(\beta_2^\prime,\infty)
\end{array}
$$
where $\beta_2^\prime$ is selected such that
$$
\int_{\beta_1}^{\beta_2^\prime} p_X(x) dx = \int_{\beta_{12,l}}^{\beta_{12,r}} p_X(x) dx
$$
One can see that we have shifted the positions of $\mathcal{T}^{f}_{1+}$ and $\mathcal{T}^{f}_{1-}$, but kept the probabilities over the regions the same, as illustrated in Fig. \ref{fig:counterExample2}.
\begin{figure}[h]
\centering
\includegraphics[height=40mm, width = 85mm]{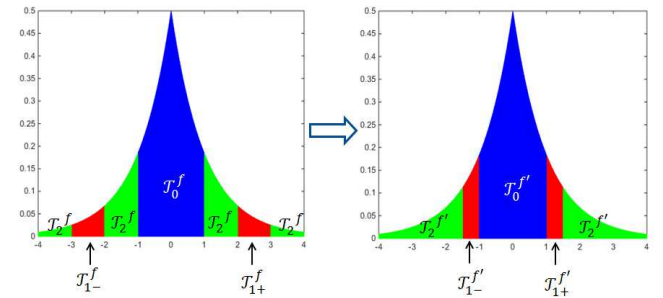}
\caption{Counter example illustrating $f'$ is no worse than $f$}
\label{fig:counterExample2}
\end{figure}
Hence, $\Prob(X\in\mathcal{T}^{f'}_{1+}) = \Prob(X\in\mathcal{T}^{f}_{1+})$, $\Prob(X\in\mathcal{T}^{f'}_{1-}) = \Prob(X\in\mathcal{T}^{f}_{1-})$, and $\Prob(X\in\mathcal{T}^{f'}_{2}) = \Prob(X\in\mathcal{T}^{f}_{2})$. Furthermore, $\mathcal{T}^{f'}_{1+}$ and $\mathcal{T}^{f'}_{1-}$ are closer to the origin than $\mathcal{T}^{f}_{1+}$ and $\mathcal{T}^{f}_{1-}$, respectively, and $p_X$ is symmetric and unimodal; hence it can be shown that
$$
\begin{array}{lcl}
\Var(X| X\in \mathcal{T}^{f'}_{1+})\leq \Var(X| X\in \mathcal{T}^f_{1+}) \vspace{0.1 cm}
\\ \Var(X| X\in \mathcal{T}^{f'}_{1-})\leq \Var(X| X\in \mathcal{T}^f_{1-})
\end{array}
$$
Therefore, $J(f')\leq J(f)$, which completes the proof.
\end{proof}
 \begin{figure*}
    \centering
      \begin{subfigure}{0.32\textwidth}
        \includegraphics[width=\textwidth]{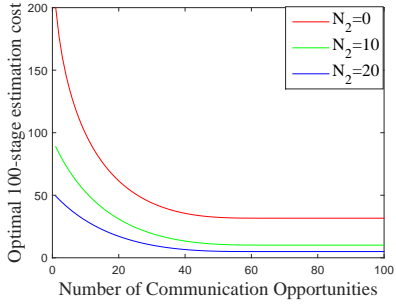}
          \caption{ 100-stage estimation error vs. $N_1$.}
          \label{fig:fixN2varyN1}
      \end{subfigure}
      \hfill
      \begin{subfigure}{0.32\textwidth}
        \includegraphics[width=\textwidth]{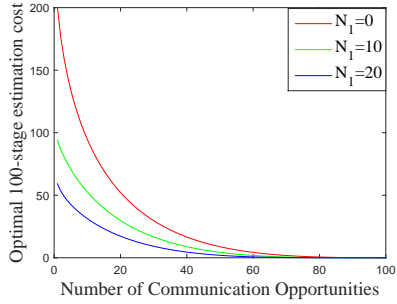}
          \caption{ 100-stage estimation error vs. $N_2$.}
          \label{fig:fixN1varyN2}
      \end{subfigure}
      \hfill
      \begin{subfigure}{0.32\textwidth}
        \includegraphics[width=\textwidth]{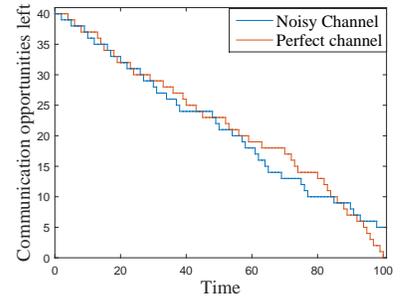}
          \caption{A sample path.}
          \label{fig:A sample path}
      \end{subfigure}

\caption{Numerical analysis.}
\label{fig:NiceImage}

\end{figure*}

With Corollary \ref{corollary for modified problem}, we reduce the optimization problem over a function space to the optimization problem over a two-dimensional space. We now compute the optimal thresholds: let $J(\beta_1,\beta_2)$ be the expected total cost corresponding to a threshold-in-threshold based communication scheduling policy $f$ with thresholds $\beta_1$ and $\beta_2$. Then \eqref{expression of J(f)} can be further computed by plugging in the expression of $f$ and applying the symmetry property of $p_X$,
$$
\begin{array}{lcl}
J(\beta_1,\beta_2) = 2\int_0^{\beta_1} x^2 p_X(x) dx + 2c_1\int_{\beta_1}^{\beta_2} p_X(x) dx + \frac{2}{\gamma+1} \vspace{0.1 cm}
\\ \cdot \Var(X|  X\in(\beta_1,\beta_2)) \int_{\beta_1}^{\beta_2} p_X(x) dx  + 2c_2 \int_{\beta_2}^{\infty} p_X(x) dx
\end{array}
$$
Taking the first derivative of $J(\beta_1,\beta_2)$ with respect to $\beta_1$, we have
$$
\begin{array}{lcl}
\frac{dJ(\beta_1,\beta_2)}{d\beta_1} = 2 \beta_1^2 \cdot p_X(\beta_1) - 2 c_1 \cdot p_X(\beta_1) \vspace{0.1 cm}
\\ + \frac{2}{\gamma+1}\cdot \frac{d}{d\beta_1}\left(\Var(X|  X\in(\beta_1,\beta_2))\int_{\beta_1}^{\beta_2} p_X(x) dx\right)
\end{array}
$$
where
$$
\begin{array}{lcl}
\dfrac{d}{d\beta_1}\left(\Var(X|  X\in(\beta_1,\beta_2))\int_{\beta_1}^{\beta_2} p_X(x)dx\right) \vspace{0.1 cm}
\\ = \dfrac{d}{d\beta_1}\left(\int_{\beta_1}^{\beta_2}x^2p_X(x)dx - \frac{\big(\int_{\beta_1}^{\beta_2}xp_X(x)dx\big)^2}{\int_{\beta_1}^{\beta_2}p_X(x)dx}  \right) \vspace{0.1 cm}
\\ =  -\beta_1^2p_X(\beta_1) + \frac{2\beta_1 p_X(\beta_1)\int_{\beta_1}^{\beta_2}xp_X(x)dx \cdot \int_{\beta_1}^{\beta_2}p_X(x)dx}{\big(\int_{\beta_1}^{\beta_2}p_X(x)dx\big)^2} \vspace{0.1 cm}
\\ \;\;\;\;- \frac{\big(\int_{\beta_1}^{\beta_2}xp_X(x)dx\big)^2 \cdot p_X(\beta_1)}{\big(\int_{\beta_1}^{\beta_2}p_X(x)dx\big)^2} \vspace{0.1 cm}
\\ = -p_X(\beta_1)\cdot\Big(\beta_1-\E[X|X\in(\beta_1,\beta_2)] \Big) ^2
\end{array}
$$
Similarly, we can compute the $\frac{d}{d\beta_2}J(\beta_1,\beta_2)$. By the first order optimality condition, the locally optimal thresholds $(\beta_1,\beta_2)$ should satisfy
\begin{equation}
\label{first order optimality}
\begin{array}{lcl}
\beta_1^2-\frac{1}{\gamma+1}\big(\beta_1-\E[X|X\in(\beta_1,\beta_2)]\big)^2- c_1 = 0 \vspace{0.1 cm}
\\ \frac{1}{\gamma+1}\big(\beta_2-\E[X|X\in(\beta_1,\beta_2)]\big)^2 + c_1-c_2 = 0 \vspace{0.1 cm}
\\ \E[X|X\in(\beta_1,\beta_2)] = \int_{\beta_1}^{\beta_2}xp_X(x)dx
\end{array}
\end{equation}
where $0\leq \beta_1 \leq \beta_2$. Once we obtain solution(s) of \eqref{first order optimality}, which are extrema of $J$,  we need to compare $J$ evaluated at the inner extrema with $J$ evaluated at the boundaries, i.e. (i) $0=\beta_1<\beta_2$, (ii) $0<\beta_1<\beta_2 = \infty$, (iii) $0\leq\beta_1=\beta_2<\infty$. The one achieving the lowest cost is the global optimal solution. Since $0\leq c_1,c_2<\infty$ and $X$ has support $\mathbb{R}$, it is easy to verify that the first two boundaries are not optimal by analyzing \eqref{three total cost functions}. Consider the third boundary $0\leq\beta_1=\beta_2<\infty$, the optimization problem collapses to the optimization problem with one perfect channel. By the results from \cite{Lipsa11}, the optimal thresholds are $\beta_1=\beta_2 = \sqrt{c_2}$. Hence, we only need to compare the performances of the inner extrema with that of $\beta_1=\beta_2 = \sqrt{c_2}$.

The existence and uniqueness of solution of \eqref{first order optimality} is not guaranteed for general parameters and densities. On the one hand, if $c_1 > c_2$, \eqref{first order optimality} does not admits a solution (see the second equation). On the other hand, since $\E[X|X\in(\beta_1,\beta_2)]$ depends on the source density $p_X$, it is hard to analyze the existence and uniqueness of the solution. For the first issue, when $c_1\geq c_2$, there is no side-effect by choosing perfect channel rather than noisy channel. Then optimization problem collapses to the optimization problem with one perfect channel. For the second case, we specify the source to have Laplace density with parameters $(0,\lambda^{-1})$, namely,
$$
p_X(x) =
\begin{cases}
\frac{1}{2}\lambda\; e^{-\lambda x},\;& x \geq 0 \vspace{0.1 cm}
\\ \frac{1}{2}\lambda\; e^{\lambda x} ,\;& x<0
\end{cases}
$$
Then $p_X$ is symmetric and unimodal. Furthermore, conditioning on $X>0$, $p_{X|X>0}$ has exponential distribution with parameter $\lambda$. Plugging for $p_X$ into \eqref{first order optimality}, and by memoryless property of exponential distribution, we have
\begin{equation}
\label{nonlinear equations we solved}
\begin{array}{lcl}
\dfrac{\Delta\beta \cdot e^{\lambda\Delta\beta}}{e^{\lambda\Delta\beta}-1} = \dfrac{1}{\lambda} + \sqrt{(c_2-c_1)(1+\gamma)} \vspace{0.1 cm}
\\ \beta_1 = \sqrt{c_1 + \dfrac{1}{1+\gamma}\Big(\Delta\beta-\sqrt{(c_2-c_1)(1+\gamma)}\Big)^2}
\end{array}
\end{equation}
where $\Delta\beta:=\beta_2-\beta_1$, $\Delta\beta > 0$. It can be verified that $\frac{\Delta\beta \cdot e^{\lambda\Delta\beta}}{e^{\lambda\Delta\beta}-1}$ is an increasing function of $\Delta\beta$ and $\frac{\Delta\beta \cdot e^{\lambda\Delta\beta}}{e^{\lambda\Delta\beta}-1}\in(\frac{1}{\lambda},\infty)$. Hence, when $c_2>c_1$, the first equation has unique solution, which uniquely determines $\beta_1$ in the second equation, and $\beta_2=\Delta\beta+\beta_1$.

\section{The Problem with Hard Constraint}
Consider the modified problem with hard constraint. Let $E_{t}^n$ and $E_{t}^p$ be the communication opportunities left at time $t$ for the noisy channel and the perfect channel, respectively. Then,
$$
E_{t}^n = N_1 - \sum_{i=1}^{t-1} \mathds{1}_{\{U_i = 1\}},\;\;\;E_{t}^p = N_2 - \sum_{i=1}^{t-1} \mathds{1}_{\{U_i = 2\}}
$$
Furthermore, let $J(t,E_{t}^n,E_{t}^p)$ be the optimal cost to go when the problem is initialized at time $t$ with $E_{t}^n$ and $E_{t}^p$ number of communication opportunities for noisy channel and perfect channel, respectively. By an argument similar to that in \cite[Theorems 2 and 3]{GaoCDC15}, the optimal decision policy at time $t$ has the form of $ U_t = f_t(X_t,E_{t}^n,E_{t}^p)$, $Y_t = g_t(X_t,S_t,E_{t}^n,E_{t}^p)$, and $\hat{X}_t = h_t(\tilde{Y}_t,S_t,E_{t}^n,E_{t}^p)$. Furthermore, the optimal cost to go $J(t,E_{t}^n,E_{t}^p)$ can be computed by solving the dynamic programming (DP) equation:
$$
\begin{array}{ccl}
&\;&J^\ast(t,E_{t}^n,E_{t}^p) \vspace{0.1 cm}
\\  &=& \underset{f_t,g_t,h_t}{\min} \{\E[(X_t-\hat{X}_t)^2]+ \E[J^\ast(t+1,E_{t+1}^n,E_{t+1}^p)]\}
\end{array}
$$
with boundary conditions $J^\ast(T+1,\cdot,\cdot) = 0$. Depending on the realization of $X_t$, $E_{t+1}^n$ may be $E_t^n$ or $E_t^n-1$, and $E_{t+1}^p$ may be $E_t^p$ or $E_t^p-1$. Hence the dynamic programming equation can be written as
$$
\label{dynamic programming equation}
\begin{array}{lcl}
J^\ast(t,E_t^n,E_t^p) = \underset{f_t,g_t,h_t}{\min} \bigg\{\E[(X_t-\hat{X}_t)^2]+J^\ast(t+1,E_t^n, \vspace{0.1 cm}
\\E_t^p)\cdot \displaystyle \int_{\mathcal{T}_0^{f_t}} p_X(x)dx + J^\ast(t+1,E_t^n-1,E_t^p) \vspace{0.1 cm}
\\ \cdot \displaystyle \int_{\mathcal{T}_1^{f_t}} p_X(x)dx  + J^\ast(t+1,E_t^n,E_t^p-1)\cdot \displaystyle \int_{\mathcal{T}_2^{f_t}} p_X(x)dx\bigg\} \vspace{0.1 cm}
\\ = J^\ast(t+1,E_t^n,E_t^p) + \underset{f_t,g_t,h_t}{\min} \bigg\{\E[(X_t-\hat{X}_t)^2]+ c_{1t}(E_t^n  \vspace{0.1 cm}
\\  ,E_t^p)\cdot \displaystyle \int_{\mathcal{T}_1^{f_t}} p_X(x)dx + c_{2t}(E_t^n,E_t^p)\cdot \displaystyle  \int_{\mathcal{T}_2^{f_t}} p_X(x)dx \bigg\} \vspace{0.1 cm}
\end{array}
$$
where $c_{1t}(E_t^n,E_t^p)=J^\ast(t+1,E_t^n-1,E_t^p)-J^\ast(t+1,E_t^n,E_t^p)$ and $c_{2t}(E_t^n,E_t^p)=J^\ast(t+1,E_t^n,E_t^p-1)-J^\ast(t+1,E_t^n,E_t^p)$. Then the problem inside $\min\{\cdot\}$ is a one stage problem with soft constraint, and communication costs are $c_{1t}(E_t^n,E_t^p)$ and $c_{2t}(E_t^n,E_t^p)$ for using the noisy channel and the perfect channel, respectively.

We assume the source process has Laplace density with parameter $(0,\lambda^{-1})$. Furthermore, we restrict sensor to apply symmetric communication scheduling strategy and encoder/decoder to apply affine encoding/decoding strategies described in Assumption 4. Then the optimal communication scheduling policy at time $t$ is threshold-in-threshold based with thresholds $(\beta_1,\beta_2)$ solved from \eqref{nonlinear equations we solved} or thresholds $(\beta'_1,\beta'_2)$ on the boundary, i.e., $\beta'_1=\beta'_2$. In order to investigate the performances of the decision policies, we solved the DP equation numerically with $\lambda=1$ and SNR $\gamma = 1$. 

In Fig. \ref{fig:fixN2varyN1}, we fix the number of communication opportunities for perfect channel, as $N_2=0,10,20$, and we  plot the optimal 100-stage estimation error versus $N_1$. When $N_2=0$, there is no perfect channel, the problem collapses to the one in \cite{GaoCDC15}. As discussed in \cite{GaoCDC15}, there exists an opportunity threshold such that the optimal 100-stage estimation error is decreasing when the number of communication opportunities is below the threshold, and staying constant above the threshold. The existence of opportunity threshold remains in the case when there is perfect channel. Furthermore, the higher $N_2$ is, the lower is the optimal 100-stage estimation error.

Fig. \ref{fig:fixN1varyN2} illustrates the performances of decision strategies when $N_1$ is fixed as $N_1=0,10,20$, and $N_2$ varies over $\{0,1,\ldots,100\}$. When $N_1=0$, there is no noisy channel, the problem collapses to the problem with perfect channel, and the plot recovers the result in \cite{Imer10}. As shown in \cite{Imer10}, the optimal estimation costs over the time horizon decreases to $0$ as the number of communication opportunities for perfect channel increases to $100$, which is the length of the time horizon. This trend remains for the case when there is noisy channel. Moreover, the more communication opportunities sensor has for the noisy channel, the lower is the optimal 100-stage estimation error.

Fig. \ref{fig:A sample path} depicts a sample path illustrating the evolution of the numbers of communication opportunities over time horizon. When generating the plot, we chose the initial numbers of communication opportunities for noisy channel $N_1=40$, and that of perfect channel, $N_2=40$. One can see that by the end of time horizon, the sensor used up all the communication opportunities for the perfect channel, but not all the communication opportunities for the noisy channel. This surprising result is due to the fact that the thresholding information, that is, whether the source realization belongs a certain interval or not, can be more informative than a noisy output from the communication channel. More discussions on interpretations of a similar result can be found in \cite[Remark 3]{GaoCDC15}.

\section{Conclusions}
In this paper, we have analyzed the impact of an additional noisy communication channel over the classical remote estimation problems. We have shown that while the intuitive solution of applying  threshold-in-threshold transmission policy may be suboptimal for the original problem, it will be optimal, under some assumptions, for the setting with a side channel. We have determined optimal policies numerically for both hard and soft constrained problems. The numerical solutions exhibit several interesting properties that are inherited from the noisy and noiseless settings.

\bibliographystyle{unsrt}
\bibliography{references}

\end{document}